\documentclass[10pt]{llncs}

\usepackage{amsfonts}

\usepackage{algorithm}
\usepackage{algorithmic}
\usepackage{graphicx}
\usepackage{wrapfig}
\graphicspath{{Figures/}}

\newcommand{\exist}{{\sc Exist}}
\newcommand{\GBWT}{\mathit{GBWT}}
\newcommand{\B}{\mathit{B}}

\usepackage{tikz}
\usetikzlibrary{arrows,backgrounds}

\makeindex
\begin{document}

\title{Computing Lempel-Ziv Factorization Online}

\institute{Lomonosov Moscow State University, Moscow, Russia\\ tat.starikovskaya@gmail.com}

\author
{
   Tatiana Starikovskaya
}
\date{\empty}
\maketitle

\begin{abstract}
We present an algorithm which computes the Lempel-Ziv factorization of a word $W$ of length $n$ on an alphabet $\Sigma$ of size $\sigma$ online in the following sense: it reads $W$ starting from the left, and, after reading each $r = O(\log_{\sigma}{n})$ characters of $W$, updates the Lempel-Ziv factorization. The algorithm requires $O(n\log\sigma)$ bits of space and $O(n \log^2{n})$ time. The basis of the algorithm is a sparse suffix tree combined with wavelet trees.
\end{abstract}

\section{Introduction}
The Lempel-Ziv factorization (further LZ-factorization for short) of a word $W$ is a decomposition $W = f_1 f_2 \ldots f_z$, where a factor $f_i$, $1 \leq i \leq z$, is either a character that does not occur in $f_1 f_2 \ldots f_{i-1}$ or the longest prefix of $f_i\ldots f_z$ that occurs in $f_1f_2 \ldots f_i$ at least twice~\cite{Crochemore:1986:TR:21537.21539,Ziv77auniversal}.

The most famous application of the LZ-factorization is data compression (e.g. the LZ-factorization is used in gzip, WinZip, and PKZIP). Moreover, it is a basis of several algorithms~\cite{Kolpakov99findingmaximal,Gusfield:2004:LTA:1046081.1046083} and text indexes~\cite{Kreft:2011:SBL:2018243.2018251}.

Let $W$ be a word of length $n$ on an alphabet $\Sigma$ of size $\sigma$. There are many algorithms that compute the LZ-factorization in $O(n \log{n})$ bits of space~\footnote{In this paper $\log$ stands for $\log_2$.}. These algorithms use suffix trees~\cite{Rodeh:1981:LAD:322234.322237}, suffix automata~\cite{Crochemore:1986:TR:21537.21539} or suffix arrays~\cite{Abouelhoda:2004:RST:985384.985389,journals/mics/ChenPS08,Crochemore:2008:CLP:1346353.1346507,conf/iwoca/CrochemoreIIKRW09,Crochemore:2008:SAC:1395764.1395899,Ohlebusch:2011:LFR:2018243.2018249} as a basis.

However, only two algorithms have been known which use $O(n \log \sigma)$ bits of space~\cite{Okanohara:2008:OAF:1431008.1431068,Ohlebusch:2011:LFR:2018243.2018249}. The algorithms exploit similar ideas (both are based on an FM-index and a compressed suffix array). The algorithm~\cite{Ohlebusch:2011:LFR:2018243.2018249} is offline and requires $O(n)$ time.

Running time of the algorithm~\cite{Okanohara:2008:OAF:1431008.1431068} is rather big, $O(n \log^3 n)$, but the algorithm computes the LZ-factorization of a word $W$ online. Consider the factors $f_1, f_2, \ldots, f_i$ of the LZ-factorization of a word $X$. The LZ-factorization of a word $Xa$, where $a$ is a character, contains either $i$ or $i+1$ factors: in the first case the factors are $f_1, f_2, \ldots, f_{i-1}, f'_i$, where the last factor $f'_i = f_i a$; and in the second case the factors are $f_1, f_2, \ldots, f_i, f_{i+1}$, where $f_{i+1} = a$. The algorithm reads $W$ and after reading each new character updates the LZ-factorization, i.e. either increases the length of the last factor by one or adds a new factor.

For many practical applications dealing with large volumes of data it would be natural to allow updating the LZ-factorization only each $r > 1$ new characters of $W$, for some small parameter $r$, in order to reduce the running time. Unfortunately, naive application of this idea to the algorithm~\cite{Okanohara:2008:OAF:1431008.1431068} does not improve its running time.

Here we propose a new linear-space algorithm which achieves a reasonable trade-off between frequency of updates and running time. The algorithm updates the LZ-factorization of $W$ each $r = \frac{\log_{\sigma}n}{4}$ characters of $W$, requiring $O(n \log^2{n})$ time and $O(n \log \sigma)$ bits of space. It is assumed that both $\sigma$ and $n$ are known beforehand and $n \geq \sigma$. The basis of the algorithm is a sparse suffix tree combined with wavelet trees.

Let $X$ be a word of length $|X|$ on $\Sigma$. Throughout the paper, positions in $X$ are numbered from~$1$. The subword of $X$ from position $i$ to position $j$ (inclusively) is denoted by $X[i..j]$. If $j = |X|$, then we write $X[i..]$ instead of $X[i..|X|]$. A word $X[i..]$ is called a suffix of $X$ and a word $X[1..j]$ is called a prefix of $X$.

With each word $Y$ of length $r$ on $\Sigma$ we associate a meta-character $Y'$ formed by concatenating bit representations of characters of $Y$. Note that a bit representation of any character of $Y$ can be obtained from the bit representation of $Y'$ in constant time by two shift operations. Also, $Y'$ can be obtained from $Y$ in $O(r)$ time.

\section{Algorithm}
Let $f_1, f_2, \ldots, f_z$ be the factors of the LZ-factorization of $W$. For the sake of clarity we describe not how to update the LZ-factorization after reading each block of characters but rather how to compute $f_1, f_2, \ldots, f_z$ sequentially. However, it will be easy to see that the presented algorithm can be modified to solve the problem we formulated in the introduction.

Suppose that $f_1, f_2, \ldots, f_{i-1}$ of total length $\ell_i$ have been computed. The algorithm consists of two procedures. The procedure $P_{<r}$ checks if $|f_i|$ is less than $r$ and, if it is, computes $f_i$ (Section~\ref{sec:<r}). The procedure $P_{\geq r}$ computes $f_i$ only if it is already known that $|f_i| \geq r$ (Section~\ref{sec:>r}). To compute $f_i$ the algorithm runs $P_{<r}$ first and then, if necessary, runs $P_{\geq r}$.

\subsection{Data Structures}
The algorithm makes use of several data structures. To explain what these data structures are, we need to give a definition of a trie and a compacted trie first.

\begin{definition}
A trie for a set of words $S$ is a rooted tree edges of which are labelled by characters. For each prefix $P$ of a word $\in S$ there exists exactly one vertex such that $P$ is spelled out by the path from the root of the trie to this vertex, and vice versa, a word spelled out by any path starting at the root must be a prefix of one of the words $\in S$. A compacted trie for $S$ can be constructed from the trie by eliminating all vertices with one son, thus forming edges that are labelled by words rather than single characters.
\end{definition}

The algorithm reads $W$ by blocks of $r$ characters starting from the left. After reading the $t$-th block of $W$, the first data structure is an (uncompacted) trie on suffixes of words $W[rj+1..r(j+2)]$, $j = 0..t-2$. Each explicit vertex $v$ of the trie stores the leftmost starting position of a suffix ending in the subtree rooted at $v$.

Let $W'$ be the meta-word formed by replacing each block of characters of $W$ with the corresponding meta-character. The second data structure is an implicit suffix tree for $W'[1..t]$, i.e. a compacted trie for the set of suffixes of $W'[1..t]$. This tree is also called a sparse suffix tree for $W[1..tr]$~\cite{ChienHonShahVitterDCC08,HonShahThanVitterSPIRE09,ChiuHonShahVitterDCC10}, though the original definition of a sparse suffix tree is slightly different~\cite{KU-96}. 

For each explicit vertex $v$ of the suffix tree we store a compacted trie $CT_v$ on words of length $r$ corresponding to the first meta-characters on the edges outgoing from $v$.

\begin{definition}
   Consider a tree with labels on edges (a suffix tree or a trie). We say that a word $X$ is \emph{represented} by a vertex $v$ (or that $v$ \emph{represents} $X$), if the word spelled out by the path from the root of the tree to $v$ is equal to $X$.
\end{definition}

If the label of an edge $(v,u)$ of the suffix tree begins with a meta-character $Y'$, and $Y$ is the corresponding word of length $r$, then we store a pointer to the edge $(v,u)$ in the leaf of $CT_v$ representing $Y$. Tries in vertices are used for navigation in the suffix tree (but not only for it). Clearly, given a vertex $v$ and a meta-character $Y'$, it takes $O(r\log{\sigma}) = O(\log{n})$ time to find an edge $(v,u)$ such that its label starts with $Y'$.

Also, the algorithm maintains a dynamic data structure which allows, given a vertex $v$ of the suffix tree, to compute the ranks of the leftmost and the rightmost leaves of the subtree rooted at $v$ in $O(\log n)$ time. 

\begin{definition}
	Block borders are positions of $W$ of the form $pr + 1$, $p = 1..\left\lfloor\frac{n}{r}\right\rfloor$.
\end{definition}

Finally, we store a data structure which allows, given an interval $I$, a word $Y \in \Sigma^{[1,r]}$, and a block border $b$, to determine whether a set of block borders corresponding to the starting positions of the suffixes which are represented by the leaves of the suffix tree with ranks in the interval $I$ contains a block border different from $b$ and preceded by an occurrence of $Y$. The procedure \exist$(I, Y, b)$ returns zero if there is no such block border and one of them otherwise.

Details of implementation are not important to understand the algorithm and will be explained later, in Section~\ref{sec:DS}.

Hereafter $\left\lfloor\frac{\ell_i}{r}\right\rfloor$ is denoted by $\ell_i'$. We assume that the algorithm has read the first $\ell_i'+1$ blocks of $W$ before running the procedures $P_{<r}$ and $P_{\geq r}$.

\subsection{Procedure $P_{<r}$}
\label{sec:<r}
Let $W[\ell_i+1..\ell_i+s]$ be the longest prefix of $W[\ell_i+1..\ell_i+r]$ which occurs before the position $\ell_i+1$. Obviously, $|f_i| \geq r$ if $s = r$ (see Fig.~\ref{fig:LM1}), and $|f_i| = s$  if $s < r$.

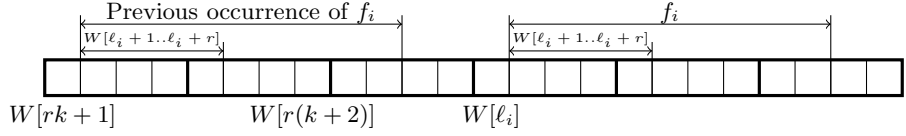
\begin{figure}[h!]
    \begin{center}
	  \begin{tikzpicture}[scale=0.95]
    \filldraw[very thick, fill=white] (0,0) rectangle (12,0.5);
    
    \foreach \x in {0.5,1,...,11.5} {
    	\draw (\x,0) -- (\x,0.5);
    }
    
    \foreach \x in {2,4,6,8,10} {
    	\draw[very thick] (\x,0) -- (\x,0.5);
    }
    
    \node[below] at (0.25,0) {\small{$W[rk+1]$}};
    \node[below] at (3.75,0) {\small{$W[r(k+2)]$}};
    \node[below] at (6.25,0) {\small{$W[\ell_i]$}};

    \draw (0.5,0) -- (0.5,1.125);
    \draw (5,0) -- (5,1.125);
    \draw (2.5,0) -- (2.5,0.75);
    
    \draw[<->] (0.5,1) -- (5,1)
    	node[pos=0.5,yshift=0.15cm]{\small{Previous occurrence of $f_i$}};
    \draw[<->] (0.5,0.625) -- (2.5,0.625)
    	node[pos=0.5,yshift=0.15cm]{\tiny{$W[\ell_i+1..\ell_i+r]$}};

    \draw (6.5,0) -- (6.5,1.125);
    \draw (11,0) -- (11,1.125);
    \draw (8.5,0) -- (8.5,0.75);
    \draw[<->] (6.5,1) -- (11,1)
    	node[pos=0.5,yshift=0.15cm]{\small{$f_i$}};
    \draw[<->] (6.5,0.625) -- (8.5,0.625)
    	node[pos=0.5,yshift=0.15cm]{\tiny{$W[\ell_i+1..\ell_i+r]$}};
\end{tikzpicture}
      \caption{Case $|f_i| \geq r$, $r = 4$. Block borders are in bold.}
      \label{fig:LM1}
    \end{center}
\end{figure}

$P_{<r}$ first computes the longest prefix $W[\ell_i+1..\ell_i+s_0]$ of $W[\ell_i+1..r(\ell_i'+1)]$ which occurs before the position $\ell_i+1$. It traverses the trie starting at the root and following edges labelled by the characters of $W[\ell_i+1..r(\ell_i'+1)]$. The algorithm stops in a vertex $v_0$ either when $v_0$ has no outgoing edge labelled by the next character of $W[\ell_i+1..r(\ell_i'+1)]$ or when the position stored in the next vertex is bigger than $\ell_i$ (which means that its label does not occur at positions $1..\ell_i$).

Clearly, $v_0$ will be labelled by $W[\ell_i+1..\ell_i+s_0]$. If $\ell_i+s_0 < r(\ell_i'+1)$, then $|f_i| = s = s_0$. Otherwise, the algorithm reads $W[r(\ell_i'+1)..r(\ell_i'+2)]$, updates the data structures and proceeds the traverse in a similar manner this time starting at $v_0$ and following edges labelled by the characters of $W[r(\ell_i'+1)+1..\ell_i+r]$. A vertex $v$ the procedure will stop at will be labelled by $W[\ell_i+1..\ell_i+s]$.

From the definition of a trie it follows that the traverse will take $O(|f_i|\log \sigma)$ time.

\subsection{Procedure $P_{\geq r}$}
\label{sec:>r}
$P_{\geq r}$ consists of two steps. The first can be considered as preliminary, and during the second step we compute $|f_i|$. 

\subsubsection{The First Step}
$P_{\geq r}$ starts with reading $W$. After reading the $s$-th block, it updates the data structures and checks whether $W'[\ell_i'+1..s]$ is represented by a leaf of the suffix tree of $W'[1..s]$. If it is, $P_{\geq r}$ proceeds to the second step. From the definition of a suffix tree it follows that after the first step of $P_{\geq r}$ all suffixes starting at positions less than $\ell_i'$ will be represented by leaves.

\begin{lemma}
\label{lm:prelim}
During the first step at most $|f_i| + r$ characters of $W$ will be read.
\end{lemma}
\begin{proof}
Since $s$ is the minimal position such that $W'[\ell_i'+1..s]$ is represented by a leaf, $W'[\ell_i'+1..s-1]$ is represented by an inner vertex in the suffix tree of $W'[1..s-1]$ and, consequently, occurs before the position $\ell_i'+1$ in $W'$. Therefore, $W[\ell_i+1..(s-1)r]$ occurs before the position $\ell_i+1$ (see Fig.~\ref{fig:LM2}) and $|f_i| \geq |W[\ell_i+1..(s-1)r]|$. The statement of the lemma easily follows.
\end{proof}

\begin{figure}[h!]
    \begin{center}
      \begin{tikzpicture}[scale=1.1]
    \filldraw[very thick, fill=white] (0,0) rectangle (12,0.5);
    
    \foreach \x in {0.5,1,...,11.5} {
    	\draw (\x,0) -- (\x,0.5);
    }
    
    \foreach \x in {2,4,6,8,10} {
    	\draw[very thick] (\x,0) -- (\x,0.5);
    }
    
    \node[below] at (6.25,0) {\small{$W[\ell_i]$}};
	\node[below] at (11.75,0) {\small{$W[rs]$}};

    \draw (0.5,0) -- (0.5,0.75);
    \draw (4,0) -- (4,0.75);
    \draw[<->] (0.5,0.625) -- (4,0.625)
    	node[pos=0.5,yshift=0.2cm]{\small{$W[\ell_i+1..r(s-1)]$}};
    
    \draw (6.5,0) -- (6.5,0.75);
    \draw (10,0) -- (10,0.75);
    \draw[<->] (6.5,0.625) -- (10,0.625)
    	node[pos=0.5,yshift=0.2cm]{\small{$W[\ell_i+1..r(s-1)]$}};

   \filldraw[very thick, fill=white] (0,-1.5) rectangle (12,-1);
             
   \foreach \x in {2,4,6,8,10} {
   	\draw[very thick] (\x,-1.5) -- (\x,-1);
   }
       
   \node[below] at (7,-1.5) {\small{$W'[\ell_i'+1]$}};
   \node[below] at (11,-1.5) {\small{$W'[s]$}};
   
   \draw (0,-1.5) -- (0,-0.75);
   \draw (4,-1.5) -- (4,-0.75);
   \draw[<->] (0,-0.875) -- (4,-0.875)
   	node[pos=0.5,yshift=0.2cm]{\small{Prev. occ. of $W'[\ell_i'+1..s-1]$}};
       
   \draw (6,-1.5) -- (6,-0.75);
   \draw (10,-1.5) -- (10,-0.75);
   \draw[<->] (6,-0.875) -- (10,-0.875)
   	node[pos=0.5,yshift=0.2cm]{\small{$W'[\ell_i'+1..s-1]$}}; 
\end{tikzpicture}
      \caption{Relation between $W'[\ell_i'+1..s-1]$ and $W[\ell_i+1..r(s-1)]$.}
      \label{fig:LM2}
    \end{center}
\end{figure}
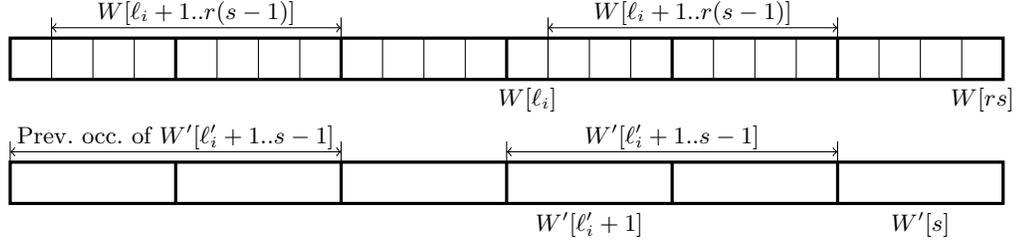

We initialize $M$ with $|W[\ell_i+1..(s-1)r]|$. From the proof of the lemma it follows that $|f_i| \geq M$. During the computation process we will increase $M$ until, finally, it will become equal to $|f_i|$.

Furthermore, the lemma guarantees that after the first step the difference between the position of the last read character of $W$ and $\ell_i + M$ is less than $r$. This invariant will be maintained throughout the second step of the procedure as well in the following way: we will read a new block of characters and update the data structures only when $\ell_i + M$ is equal to the position of the last read character of $W$.

\subsubsection{The Second Step}
Consider the first block border which intersects a previous occurrence of $f_i$ (see Fig.~\ref{fig:LM3}). It divides the occurrence into two parts: the first short part equal to $W[\ell_i+1..\ell_i+m-1]$ and the second part equal to a prefix of $W[\ell_i+m..]$, $m \in [1,r]$.

\begin{figure}[h!]
    \begin{center}
      \begin{tikzpicture}[scale=0.95]
    \filldraw[very thick, fill=white] (0,0) rectangle (12,0.5);
    \filldraw[fill=gray!20] (0.5,0) rectangle (2,0.5);
    \draw[very thick] (0.5,0) -- (2,0);
    \draw[very thick] (0.5,0.5) -- (2,0.5);
    
    \foreach \x in {0.5,1,...,11.5} {
    	\draw (\x,0) -- (\x,0.5);
    }
    
    \foreach \x in {2,4,6,8,10} {
    	\draw[very thick] (\x,0) -- (\x,0.5);
    }
    
    \node[below] at (0.25,0) {\small{$W[rk+1]$}};
    \node[below] at (3.75,0) {\small{$W[r(k+2)]$}};
    \node[below] at (6.25,0) {\small{$W[\ell_i]$}};

    \draw (0.5,0) -- (0.5,0.75);
    \draw (5,0) -- (5,0.75);
    
    \draw[<->] (0.5,0.625) -- (5,0.625)
    	node[pos=0.5,yshift=0.15cm]{\small{Previous occurrence of $f_i$}};

    \draw (6.5,0) -- (6.5,0.75);
    \draw (11,0) -- (11,0.75);
    \draw[<->] (6.5,0.625) -- (11,0.625)
    	node[pos=0.5,yshift=0.15cm]{\small{$f_i$}};
\end{tikzpicture}
      \caption{A previous occurrence of $f_i$. The part equal to $W[\ell_i+1..\ell_i+m-1]$ ($m = 4$) is highlighted in grey.}
      \label{fig:LM3}
    \end{center}
\end{figure}
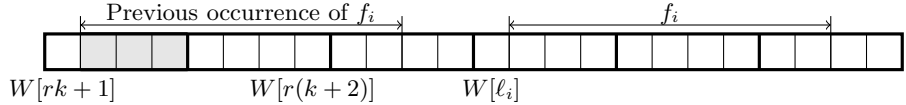

Let $f_i^m$ be the longest prefix of $W[\ell_i+m..]$ with at least one occurrence at a block border which is less than $\ell_i+1$ and preceded by an occurrence of $W[\ell_i+1..\ell_i+m-1]$. Obviously, $|f_i| = \max_{m \in [1,r]}(|f_i^m|+m-1)$. 

For each $m = 1..r$ the procedure $P_{\geq r}$ either computes $|f_i^m|$ and updates $M$ or proves that $|f_i^m|+m-1 \leq M$ and starts computation of $|f_i^{m+1}|$.

If $(\ell_i'+1)r+1-m \le \ell_i$, then the second step of $P_{\geq r}$ starts with computing the length $q$ of the longest common prefix of  $W[\ell_i+1..]$ and $W[(\ell_i'+1)r+1-m..]$. In order to compute $q$ the procedure compares the two strings character by character. If $q > M$, the procedure puts $M$ equal to $q$. Then the procedure starts to work with the suffix tree.

The procedure traverses the suffix tree starting at the root and following the edges so that characters of the words corresponding to meta-characters of labels coincide with characters of $W[\ell_i+m..]$. For navigation the procedure uses compact tries stored in the vertices of the suffix tree.

Suppose that after reading a word $W[\ell_i+m..p]$ of length at least $M-m+1$ the procedure is on the edge $(v,u)$ of the suffix tree. Two cases are possible depending on whether $u$ is an inner vertex of the suffix tree or a leaf.

Let $u$ be an inner vertex. Obviously, $|f_i^m| \geq |W[\ell_i+m..p]|$ iff a set of the block borders corresponding to the leaves of the subtree rooted at $u$ contains a block border less than $\ell_i+1$ preceded by an occurrence of $W[\ell_i+1..\ell_i+m-1]$.

\begin{definition}
  \emph{String depth} of a vertex $u$ of the suffix tree is the length of its label.
\end{definition}

\begin{lemma}
\label{lm:B_v_prop}
 Let $u$ be an explicit inner vertex of the suffix tree of $W'[1..s]$ with string depth at least $\left\lfloor\frac{M}{r}\right\rfloor$. Then a block border corresponding to a leaf in the subtree rooted at $v$ can not be bigger than $(\ell_i'+1)r + 1$.
\end{lemma}
\begin{proof}
  Indeed, a subtree rooted at $u$ can only contain leaves representing suffixes of length at least $\left\lfloor\frac{M}{r}\right\rfloor + 1 = s-\ell_i'$, and all such suffixes start at positions $\leq \ell_i'+1$. The statement immediately follows.
\end{proof}

Since $|W[\ell_i+m..p]| \geq M-m+1$, the string depth of $u$ is at least $\lceil\frac{M-m+1}{r}\rceil \geq \lfloor \frac{M}{r} \rfloor$. It follows from the lemma that a set of the block borders corresponding to the leaves of the subtree rooted at $u$ might contain only one block border situated to the left of $\ell_i+1$, namely, $(\ell_i'+1)r + 1$. 

Let $left(u)$ and $right(u)$ be the ranks of the leftmost and the rightmost leaves of the subtree rooted at $u$. The ranks $left(u)$ and $right(u)$ can be computed in $O(\log n)$ time (see Section~\ref{sec:DS}). Then the set of the block borders corresponding to the leaves of the subtree rooted at $u$ contains a block border less than $\ell_i+1$ preceded by an occurrence of $W[\ell_i+1..\ell_i+m-1]$ iff the set of block borders corresponding to leaves with ranks belonging to the interval $[left(u), right(u)]$ contains at least one block border different from $(\ell_i'+1)r + 1$ and preceded by an occurrence of $W[\ell_i+1..\ell_i+m-1]$. To define is this condition holds we call the procedure \exist{}. If such a block border exists, the procedure updates $M$ and proceeds. Otherwise, the procedure starts computation of $|f_i^{m+1}|$.

If $u$ is a leaf then instead calling the procedure \exist{} we first check if this leaf corresponds to a block border less than $\ell_i+1$ and then check if the border is preceded by an occurrence of $W[\ell_i+1..\ell_i+m-1]$ using a character-by-character comparison. 

Suppose now that after reading $W[\ell_i+m..p]$ of length at least $M - m + 1$ the procedure stops on an edge $(v', u')$ of the compact trie $CT_v$, stored in a vertex $v$ of the suffix tree (which means that we are looking for an edge outgoing from $v$ which has an appropriate label). Let $u_1$ and $u_2$ be sons of $v$, which correspond to the leftmost and the rightmost leaves of the subtree of $CT_v$ rooted at $u'$. Obviously, $u_1$ and $u_2$ can be found in $O(r)$ time. All block borders corresponding to leaves with ranks $left(u_1), left(u_1) + 1, \ldots, right(u_2)$ are the starting positions of occurrences of $W[\ell_i+m..p]$. Moreover, if one of these block borders is bigger than $\ell_i+1$ then it is equal to $(\ell_i'+1)r + 1$ (Lemma~\ref{lm:B_v_prop}). To define if there is a block border corresponding to a leaf with the rank in the interval $[left(u_1), right(u_2)]$ preceded by an occurrence of $W[\ell_i..\ell_i+m-1]$ and different from $(\ell_i'+1)r + 1$ the algorithm calls the procedure \exist{}.

Correctness of the procedure $P_{\geq r}$ follows from its description. The following lemma estimates the time spent during $P_{\geq r}$, not including the time for updates of the data structures.

\begin{lemma}
\label{lm:second_step}
To compute $f_i$ the procedure $P_{\geq r}$ needs $O(|f_i|\log^2{n}+r\log^2{n})$ time.
\end{lemma}
\begin{proof}
During the first step $P_{\geq r}$ reads $O(|f_i|+r)$ characters of $W$ (Lemma~\ref{lm:prelim}).

To compute the longest common prefix of $W[\ell_i+1..]$ and $W[(\ell_i'+1)r+1-m..]$ we need $O(|f_i^m|)$ time. To follow $f_i^m$ down in the suffix tree we need $O(|f_i^m|\log \sigma)$ time. Since after each execution of the procedure \exist{} we either increase $M$ or proceed to the computation of $f_i^{m+1}$, it is executed at most $r + |f_i|$ times. The procedure \exist{} takes $O(\log^2{n})$ time (see Section~\ref{sec:DS}). Therefore, the total time spent during the second step of $P_{\geq r}$ is $O((r + |f_i|)\log^2{n} + r|f_i|\log\sigma) = O(|f_i|\log^{2}{n}+r\log^2{n})$.
\end{proof}

\section{Data Structures}
\label{sec:DS}
As we have already said, our algorithm maintains two data structures. In this section we give the details and describe update procedures.

\subsection{Trie}
After reading $W[1..tr]$ the trie contains suffixes of words $W[rj+1..r(j+2)]$, $j = 0..t-2$. To update the trie after reading the $(t+1)$-th block of characters we first check if $W[r(t-1)+1..r(t+1)]$ is represented in the trie. To do that we traverse the trie starting at the root and following edges labelled by the characters of $W[r(t-1)+1..r(t+1)]$. If we read out the whole word, then $W[r(t-1)+1..r(t+1)]$, and, consequently, all its suffixes are represented in the trie. If not, we add all suffixes of $W[r(t-1)+1..r(t+1)]$, including the word itself, to the trie. 

\begin{lemma}
\label{lm:CT}
The trie occupies $o(n)$ bits of space and its maintenance takes $O(n\log\sigma)$ time.
\end{lemma}
\begin{proof}
Due to our choice of $r$, there are at most $\sigma^{2r} = \sigma^{\frac{\log_{\sigma}{n}}{2}} = n^{\frac{1}{2}}$ different words of length $2r$ on $\Sigma$. Therefore, the trie has at most $n^{\frac{1}{2}} r^2$ vertices and occupies $o(n)$ bits of space.

To check if the words $W[rj+1..r(j+2)]$, $j = 0..\frac{n}{r}-2$, are represented in the trie one needs $O(n\log\sigma)$ time in total. During the algorithm we add suffixes of at most $n^{\frac{1}{2}} < \frac{n}{r^2}$ words. All suffixes of a word of length $2r$ can be added to the trie in $O(r^2 \log\sigma)$ time, so we get the announced time bound.
\end{proof}

Finally, suppose that we create a new vertex $v$ in the process of adding a suffix $W[p..r(k+2)]$ of the word $W[rk+1..r(k+2)]$ to the trie. Then we just remember the position $p$ as the leftmost starting position of a suffix ending in the subtree rooted at $v$. This completes the description of the update procedure of the trie.

\subsection{Suffix Tree}
The suffix tree is updated by Ukkonen's algorithm~\cite{ukkonen:on-line}. When we create a new edge outgoing from a vertex $v$ with the first character of the label equal to $W'[k]$, we add $W[(k-1)r+1..kr]$ to $CT_v$.

Below we describe the procedure \exist{} and how to compute the ranks of the leftmost and the rightmost leaves in a subtree rooted at a vertex $v$.

\subsubsection{Ranks of the leftmost and the rightmost leaves}
The data structure we will use to compute the ranks of the leftmost and the rightmost leaves of a subtree is similar to the one from~\cite{KNS12}.

We maintain a dynamic doubly-linked list $EL$ corresponding to the Euler tour of the current suffix tree. Each internal vertex of the suffix tree is stored in two copies in $EL$, corresponding respectively to the first and last visits of the vertex during the Euler tour. Leaves of the suffix tree are kept in one copy. Observe that the leaves of the suffix tree appear in $EL$ in the ``left-to-right'' order, although not consecutively.

We also maintain a balanced binary tree, denoted $BT$, whose leaves are elements of $EL$. Note that the number of vertices of $BT$ is bounded by $2\frac{n}{r}$ and the height of~$BT$ is $O(\log n)$. We call leaves of $BT$ corresponding to leaves of the suffix tree {\em suffix leaves}. For each suffix leaf we store the corresponding block border (we will use this information in the procedure \exist{}), and for each internal vertex $u$ of~$BT$ we store the number of suffix leaves in the subtree of~$BT$ rooted at~$u$.

The rank of the leftmost leaf in the subtree rooted at $v$ is the number of the suffix leaves in $EL$ preceding the first copy of $v$ in $EL$ plus one. This number can be computed in $O(\log{n})$ time by following the path from the leaf of $BT$ corresponding to this copy to the root of $BT$ and summing up the number of the suffix leaves in the subtrees rooted at the left sons of the vertices on the path. The rank of the rightmost leaf can be computed in a similar way.

Now we should explain how to update $EL$ and $BT$. When a new vertex $v$ is added to a suffix tree, the following updates should be done (in order):

\begin{itemize}
\item[(i)] insert $v$ at the right place of the list $EL$ (in two copies if $v$ is an internal vertex),
\item[(ii)] rebalance the tree $BT$ if needed,
\item[(iii)] if $v$ is a leaf of the suffix tree (i.e. a suffix leaf of $BT$), update information about the number of suffix leaves in $BT$.
\end{itemize}

To see how update (i) works, we have to recall how suffix tree is updated when a new document is inserted. Two possible updates are creation of a new internal vertex $v$ by splitting an edge into two (edge subdivision) and creating a new leaf $u$ as a child of an existing vertex. In the first case, we insert the first copy of $v$ right after the first copy of its parent, and the second copy right before the second copy of its parent. In the second case, the parent of $u$ has already at least one child, and we insert $u$ either right after the second (or the only) copy of its left sibling, or right before the first (or the only) copy of its right sibling. 

Rebalancing the tree $BT$ (update (ii)) is done using standard methods. Observe that during the rebalancing we may have to adjust the information about the number of the suffix leaves for internal vertices, but this is easy to do as only a constant number of local modifications is done at each level.

Update (iii) is triggered when a new leaf $u$ is created in the suffix tree and added to $EL$. We then have to follow the path in $BT$ from the new leaf $u$ to the root and update the information about the number of suffix leaves for all vertices on this path. These updates are
straightforward. All these steps take $O(\log n)$ time. 

\subsubsection{Procedure \exist$(I, Y, b)$}
Let $p_i$ be the starting position of the suffix represented by the $i$-th leaf in the left-to-right order on the leaves of the suffix tree. Consider a virtual sequence $\GBWT$, where $\GBWT[i]$ is equal to the reverse of the bit representation of $W'[p_i-1]$. Elements of $\GBWT$ belong to a segment $[0,\sigma^{r}] = [0,n^{\frac{1}{4}}]$.

Consider a dynamic wavelet tree for $\GBWT$. The wavelet tree for a sequence $\GBWT$, elements of which belong to a segment $[min, max] \subset [0,n^{\frac{1}{4}}]$ can be defined recursively. If $min = max$ then the wavelet tree consists of one vertex corresponding to $min$. Otherwise the tree has a root corresponding to the segment $[min,max]$. A binary vector $V_{root}$ is defined as follows: if $\GBWT[i] \leq \lfloor\frac{min+max}{2}\rfloor$, then $V_{root}[i] = 0$, otherwise $V_{root}[i] = 1$. We store a data structure~\cite{Makinen:2008:DES:1367064.1367072} which allows ro read any bit $V_{root}[i]$, compute the number of zeros or ones in a prefix $V_{root}[1..i]$ ($rank_0(i, V_{root})$ or $rank_1(i,V_{root})$ ), or compute the position of $i$-th zero or $i$-th one ($select_0(i, V_{root})$ or $select_1(i,V_{root})$), as well as to add a new bit between $V_{root}[i]$ and $V_{root}[i+1]$ in $O(\log n)$ time. The data structure occupies the number of bits proportional to the vector's length. Let $\GBWT_{left}$ be a subsequence of $\GBWT$ formed by the elements $\GBWT[i]$, $\GBWT[i] \leq \lfloor\frac{min+max}{2}\rfloor$, and $\GBWT_{right}$ be the complementary subsequence. Then a subtree rooted at the left son of the root is the wavelet tree for $\GBWT_{left}$, elements of which belong to a segment $[min, \lfloor\frac{min+max}{2}\rfloor]$, and a subtree rooted at the right son of the root is the wavelet tree for $\GBWT_{right}$, elements of which belong to a segment $[\lfloor\frac{min+max}{2}\rfloor + 1, max]$.

It follows from the definition that the wavelet tree for $\GBWT$ has $o(n)$ leaves and, consequently, $o(n)$ vertices. As there are at most $O(\log n)$ levels in the tree and the total length of the bit vectors is at most $\frac{n}{r}$, the wavelet tree for $\GBWT$ occupies $o(n) + O(\frac{n}{r} \log n) = O(n \log \sigma)$ bits of space in total.

We define a meta-character $c_{min}$ as follows: reverse the bit representation of $Y$ and then append $(r-|Y|)\log{\sigma}$ zeros to it. A meta-character $c_{max}$ is defined in a similar way, but ones are appended instead of zeros. Obviously, a block border $pr+1$ is preceded by an occurrence of $Y$ iff the reverse of the bit representation of $W'[p-1]$ lies in the interval $[c_{min},c_{max}]$. Let $B^Y_I$ be the set of block borders corresponding to the leaves with ranks in $I$ and preceded by an occurrence of $Y$. First the procedure \exist{} finds $p_1, p_2 \in I$ such that $\GBWT[p_1], \GBWT[p_2] \in [c_{min},c_{max}]$, and then computes the block borders $b_1, b_2$ corresponding to leaves with ranks $p_1, p_2$. Obviously, $\B^Y_I$ contains at least one block border different from $b$ iff either $b_1$ or $b_2$ is not equal to $b$.

The dynamic wavelet tree for $\GBWT$ allows to find $p_1, p_2 \in I$ such that $\GBWT[p_1], \GBWT[p_2]$ belong to $[c_{min},c_{max}]$ in $O(\log^2 n)$ time~\cite{Makinen06position-restrictedsubstring}. We start at the interval $[start, end] = I$ of $V_{root}$. Now we map the interval to the left and to the right, replacing $start$ by $rank_{0/1}(start - 1,V_{root}) + 1$ and $end$ by $rank_{0/1}(end,V_{root})$, and continue recursively. We stop the recursion (i) if the interval $[start,end]$ is empty; (ii) if the interval corresponding to the current vertex does not intersect with the interval $[c_{min},c_{max}]$; (iii) if the interval corresponding to the current vertex is contained in $[c_{min},c_{max}]$. It can be shown that only $O(\log n)$ vertices will be visited. Suppose that the traverse stops at some vertex $u$ because of (iii) and $\GBWT_u$ is the corresponding subsequence of $\GBWT$. Then elements $\GBWT_u[start], \GBWT_u[start+1], \ldots, \GBWT_u[end]$ belong to the interval $[c_{min},c_{max}]$, and their positions in $\GBWT$ belong to $I$. Any of these positions, e.g., the position of an element $\GBWT_u[k]$ in $\GBWT$, can be computed in $O(log^2 n)$ time in the following way: we go along the path from $u$ to the root replacing $k$ by $select_{0/1}(k, V_p)$ when moving from the left (right) son of a vertex $p$ to $p$. The value of $k$ at the root will be equal to the position of the element $\GBWT_u[k]$ in $\GBWT$.

Using $BT$ the block borders corresponding to the leaves of the suffix tree with ranks $p_1$ and $p_2$ can be computed in $O(\log n)$ time, since it is enough to find the corresponding suffix leaves of~$BT$. Hence, $O(\log^2 n)$ time is sufficient to determine whether a set of block borders corresponding to the starting positions of the suffixes which are represented by the leaves of the suffix tree with ranks in the interval $I$ contains a block border different from $b$ and preceded by an occurrence of $Y$.

It remains to describe how the wavelet tree is updated. To add a new element between $\GBWT[i]$ and $\GBWT[i+1]$ we need $O(\log^2 n)$ time, because we need to create at most $\log n$ vertices and to add a new bit to $O(\log n)$ binary vectors. To update the wavelet tree after adding a new leaf to the suffix tree we first compute the rank of this leaf in the left-to-right order on the leaves of the suffix tree in $O(\log{n})$ time using $BT$ and then add the corresponding element to $\GBWT$.

\begin{lemma}
\label{lm:ST}
The suffix tree and additional data structures occupy $O(n \log{\sigma})$ bits and their maintenance takes $O(n \log^2{n})$ time.
\end{lemma}
\begin{proof}
The suffix tree has at most $\frac{n}{r}$ leaves and therefore $O(\frac{n}{r})$ edges. We specify labels of edges by their starting and final positions in $W'$. Hence, the suffix tree occupies $O(n \log{\sigma})$ bits.

Tries in vertices of the suffix tree have $O(\frac{n}{r})$ leaves in total and occupy $O(n \log{\sigma})$ bits as well (labels of edges are specified by their starting and final positions in $W$). Finally, $BT$, $EL$ and the dynamic wavelet tree use $O(\frac{n}{r}\log{n}) = O(n \log\sigma)$ bits of space.

Ukkonen's algorithm~\cite{ukkonen:on-line} takes $O(\frac{n}{r}\log{n}) = O(n \log\sigma)$ time (additional $\log n$ appears because of the cost of navigation). To update tries in the vertices of the suffix tree we need $O(\frac{n}{r} \log n) = O(n\log\sigma)$ time. All wavelet tree updates take $O(\frac{n}{r} \log^2{n}) = O(n \log{n}\log\sigma) = O(n \log^2{n})$ time. And finally, updates of $BT$ and $EL$ take $O(\frac{n}{r}\log{n}) = O(n \log\sigma)$ time.
\end{proof}

\section{Results and Conclusions}
To conclude, we prove the following theorem. 

\begin{theorem}
  The presented algorithm computes the Lempel-Ziv factorization of a word $W$ in $O(n\log^2{n})$ time and $O(n \log{\sigma})$ bits of space.
\end{theorem}
\begin{proof}
Lemmas~\ref{lm:CT} and~\ref{lm:ST} guarantee that the data structures occupy $O(n\log{\sigma})$ bits of space in total and that their maintenance takes $O(n \log^2{n})$ time.

To compute $f_i$, first $P_{<r}$ is run. As we have proved, it takes $O(|f_i| \log \sigma)$ time. $P_{\geq r}$ is run only when $|f_i| \geq r$ (i.e., at most $\frac{n}{r}$ times) and takes $O((|f_i|+r)\log^2{n})$ time. Therefore, the total time spent by procedures $P_{<r}$ and $P_{\geq r}$ is $O(n \log^2 n)$, and this completes the proof.
\end{proof}

It is easy to see that the described algorithm can be implemented online with the same running time and space.

\paragraph{Acknowledgement} 
The author has been supported by a grant 10-01-93109-CNRS-a of the Russian Foundation for Basic Research and by the mobility grant funded by the French Ministry of Foreign Affairs through the EGIDE agency.

The author thanks Simon J. Puglisi, who proposed to consider the problem of computation of the Lempel-Ziv factorization on a sparse suffix tree, and Gregory Kucherov and Alexei Lvovich Semenov for very helpful discussions.

\bibliographystyle{plain}
\bibliography{main}

\begin{thebibliography}{10}

\bibitem{Abouelhoda:2004:RST:985384.985389}
Mohamed~Ibrahim Abouelhoda, Stefan Kurtz, and Enno Ohlebusch.
\newblock Replacing suffix trees with enhanced suffix arrays.
\newblock {\em J. of Discrete Algorithms}, 2:53--86, 2004.

\bibitem{journals/mics/ChenPS08}
Gang Chen, Simon~J. Puglisi, and William~F. Smyth.
\newblock {L}empel-{Z}iv factorization using less time \& space.
\newblock {\em Mathematics in Computer Science}, 1(4):605--623, 2008.

\bibitem{ChienHonShahVitterDCC08}
Shing-Yan Chiu, Wing-Kai Hon, Rahul Shah, and Jeffrey~S. Vitter.
\newblock Geometric {B}urrows-{W}heeler transform: Linking range searching and
  text indexing.
\newblock In {\em Proceedings of the Data Compression Conference}, pages
  252--261. IEEE Computer Society, 2008.

\bibitem{ChiuHonShahVitterDCC10}
Shing-Yan Chiu, Wing-Kai Hon, Rahul Shah, and Jeffrey~S. Vitter.
\newblock {I/O}-efficient compressed text indexes: From theory to practice.
\newblock In {\em Proceedings of the Data Compression Conference}, pages
  426--434. IEEE Computer Society, 2010.

\bibitem{HonShahThanVitterSPIRE09}
Shing-Yan Chiu, Rahul Shah, Sharma~V. Thankachan, and Jeffrey~S. Vitter.
\newblock On entropy-compressed text indexing in external memory.
\newblock In {\em Proceedings of the 16th International Symposium on String
  Processing and Information Retrieval}, volume 5721 of {\em Lecture Notes in
  Computer Science}, pages 75--89. Springer, 2009.

\bibitem{Crochemore:1986:TR:21537.21539}
Maxime Crochemore.
\newblock Transducers and repetitions.
\newblock {\em Theor. Comput. Sci.}, 45:63--86, 1986.

\bibitem{Crochemore:2008:CLP:1346353.1346507}
Maxime Crochemore and Lucian Ilie.
\newblock Computing longest previous factor in linear time and applications.
\newblock {\em Inf. Process. Lett.}, 106:75--80, 2008.

\bibitem{conf/iwoca/CrochemoreIIKRW09}
Maxime Crochemore, Lucian Ilie, Costas~S. Iliopoulos, Marcin Kubica, Wojciech
  Rytter, and Tomasz Walen.
\newblock {LPF} computation revisited.
\newblock In {\em IWOCA}, volume 5874 of {\em Lecture Notes in Computer
  Science}, pages 158--169. Springer, 2009.

\bibitem{Crochemore:2008:SAC:1395764.1395899}
Maxime Crochemore, Lucian Ilie, and William~F. Smyth.
\newblock A simple algorithm for computing the {L}empel {Z}iv factorization.
\newblock In {\em Proceedings of the Data Compression Conference}, pages
  482--488, Washington, DC, USA, 2008. IEEE Computer Society.

\bibitem{Gusfield:2004:LTA:1046081.1046083}
Dan Gusfield and Jens Stoye.
\newblock Linear time algorithms for finding and representing all the tandem
  repeats in a string.
\newblock {\em J. Comput. Syst. Sci.}, 69:525--546, 2004.

\bibitem{KU-96}
Juha K{\"a}rkk{\"a}inen and Esko Ukkonen.
\newblock Sparse suffix trees.
\newblock In {\em Proceedings of the 2nd Annual International Computing and
  Combinatorics Conference}, volume 1090 of {\em Lecture Notes in Computer
  Science}, pages 219--230. Springer Verlag, 1996.

\bibitem{Kolpakov99findingmaximal}
Roman Kolpakov and Gregory Kucherov.
\newblock Finding maximal repetitions in a word in linear time.
\newblock In {\em Proceedings of the 1999 Symposium on Foundations of Computer
  Science}, pages 596--604. IEEE Computer Society, 1999.

\bibitem{Kreft:2011:SBL:2018243.2018251}
Sebastian Kreft and Gonzalo Navarro.
\newblock Self-indexing based on {LZ77}.
\newblock In {\em Proceedings of the 22nd annual conference on Combinatorial
  Pattern Matching}, CPM'11, pages 41--54, Berlin, Heidelberg, 2011.
  Springer-Verlag.

\bibitem{KNS12}
G.~Kucherov, Y.~Nekrich, and T.~Starikovskaya.
\newblock Cross-document pattern matching.
\newblock In {\em Proceedings of the 23rd Annual Symposium on Combinatorial
  Pattern Matching (CPM), July 3-5, 2012, Helsinki (Finland)}, Lecture Notes in
  Computer Science. Springer Verlag, 2012.
\newblock to appear.

\bibitem{Makinen06position-restrictedsubstring}
Veli M\"{a}kinen and Gonzalo Navarro.
\newblock Position-restricted substring searching.
\newblock In {\em Proceedings of the 7th Latin American Symposium}, Lecture
  Notes in Computer Science, pages 703--714. Springer, 2006.

\bibitem{Makinen:2008:DES:1367064.1367072}
Veli M\"{a}kinen and Gonzalo Navarro.
\newblock Dynamic entropy-compressed sequences and full-text indexes.
\newblock {\em ACM Trans. Algorithms}, 4:32:1--32:38, 2008.

\bibitem{Ohlebusch:2011:LFR:2018243.2018249}
Enno Ohlebusch and Simon Gog.
\newblock {L}empel-{Z}iv factorization revisited.
\newblock In {\em Proceedings of the 22nd annual conference on Combinatorial
  Pattern Matching}, CPM'11, pages 15--26, Berlin, Heidelberg, 2011.
  Springer-Verlag.

\bibitem{Okanohara:2008:OAF:1431008.1431068}
Daisuke Okanohara and Kunihiko Sadakane.
\newblock An online algorithm for finding the longest previous factors.
\newblock In {\em Proceedings of the 16th Annual European Symposium on
  Algorithms}, ESA '08, pages 696--707, Berlin, Heidelberg, 2008.
  Springer-Verlag.

\bibitem{Rodeh:1981:LAD:322234.322237}
Michael Rodeh, Vaughan~R. Pratt, and Shimon Even.
\newblock Linear algorithm for data compression via string matching.
\newblock {\em J. ACM}, 28:16--24, 1981.

\bibitem{ukkonen:on-line}
Esko Ukkonen.
\newblock On-line construction of suffix trees.
\newblock {\em Algorithmica}, pages 249--260, 1995.

\bibitem{Ziv77auniversal}
Jacob Ziv and Abraham Lempel.
\newblock A universal algorithm for sequential data compression.
\newblock {\em IEEE Transactions on Information Theory}, 23(3):337--343, 1977.

\end{thebibliography}

\end{document}